\documentclass[a4paper, aps, prl,  singlecolumn,groupedaddress]{revtex4}
\usepackage[utf8]{inputenc}
\usepackage[colorlinks=true,linkcolor=blue,citecolor=blue,
            urlcolor=blue]{hyperref}
\usepackage{graphicx}
\usepackage{amsmath}
\usepackage{amsthm}
\usepackage{amsfonts}
\usepackage{rotating}

\usepackage{float}

\newcommand{\bea}{\begin{eqnarray}}
\newcommand{\eea}{\end{eqnarray}}

\def\bi{\begin{itemize}}
\def\ei{\end{itemize}}
\def\bc{\begin{center}}
\def\ec{\end{center}}

\def\C{\hbox{$\mit I$\kern-.7em$\mit C$}}
\def\R{\hbox{$\mit I$\kern-.6em$\mit R$}}
\def\N{\hbox{$\mit I$\kern-.6em$\mit N$}}

\newtheorem{theorem}{Theorem}

\newtheorem{lemma}[theorem]{Lemma}

\newtheorem{observation}[theorem]{Observation}

\begin{document}
\title{Signaling between time steps does not allow for nonlocality beyond hidden nonlocality} 

\author{Cornelia Spee}

\affiliation{Institute for Quantum Optics and Quantum Information (IQOQI),
Austrian Academy of Sciences, Boltzmanngasse 3, 1090 Vienna, Austria}
\affiliation{Institute for Theoretical Physics, University of Innsbruck, Technikerstr. 21A, 6020 Innsbruck, Austria}

\begin{abstract}  
Hidden nonlocality is the phenomenon that entangled states can be local in the standard Bell scenario but display nonlocality after local filtering. However, there exist entangled states for which all measurement statistics can be described via a local hidden variable model even after local filtering. In this work we consider the scenario that measurement outcomes and settings of Alice can influence measurements of Bob in subsequent time steps (and vice versa), however, there is no signaling among them for measurements at the same time step.  We show that in this scenario states that only display local statistics after local filtering remain local even when considering the complete statistics of arbitrary sequences and therefore no advantage can be gained by performing longer sequences in this scenario.
We first determine the extreme points of the polytope defined by the no-signaling conditions within the same time step and the arrow-of-time constraints. Based on these results we introduce a notion of locality and provide a complete representation of the corresponding local polytope in terms of inequalities in the simplest scenario. These results imply that in the scenario considered here there is no nonlocality beyond hidden nonlocality. We further propose a device-dependent Schmidt number witness and we compare our finding to known local models in the sequential scenario.\end{abstract}
\maketitle

{\it Introduction.---}
Initially introduced as a quantitative criteria for ruling out local hidden variable models \cite{Bell} Bell inequalities  (or more generally non-local correlations) 
serve nowadays also as the basis for applications such as randomness generation \cite{randomness,randomnessgen}, quantum communication protocols \cite{E91,DIQKD} and self-testing \cite{selftesting0,selftesting}. In particular, their device-independence, i.e., no assumptions are made on the functionality of the preparation and measurement apparatuses apart from no-signaling among the measurement devices of the two parties, makes them a useful tool.
Motivated by their significance for our understanding of quantum mechanics, no-signaling correlations have been widely studied \cite{rev}. In general there are three sets of correlations to distinguish. The local correlations, which are a subset of the other two and which form the local polytope, can be explained by local hidden variable models. The quantum set, which is notoriously difficult to characterize, consists of all no-signaling correlations that can be realized within quantum mechanics.  
Finally, the no-signaling polytope, which is  only constrained by the no-signaling conditions (and the conditions that any probability distribution has to obey, i.e., positivity and normalization)  is a strict superset of the quantum set. In particular, the well-known PR boxes \cite{PR,PR2,PR3}, which are extreme points of the polytope, cannot be realized within quantum mechanics. Inequalities distinguishing the local polytope from the quantum set (Bell inequalities) are highly desirable and have been found for many different scenarios \cite{rev}, i.e., number of inputs and outputs. In particular, for simple scenarios like two measurement settings with each two outcomes, all facet inequalities of the local polytope are known \cite{allfacet,allfacets}. These correspond up to relabelling to the famous CHSH inequality \cite{CHSH} or the trivial constraints (i.e., positivity of  probabilities).  
In order to violate a Bell inequality the underlying quantum state has to be entangled. This allows one to use Bell inequalities in order to detect entanglement in a device-independent way. However, the converse is not true.There exist entangled states that do not violate any Bell inequality \cite{Werner,forPOVM}, i.e., measurements always  result  in local correlations and therefore entanglement and Bell nonlocality are two distinct features.

It has been noticed that via local filtering, i.e., by performing local measurements and post-processing on a single outcome, the Bell nonlocality of some states can be activated \cite{hiddennonlocal0,hiddennonlocal}. That is, for a single measurement round (the standard Bell scenario), the states only display local behavior, whereas for two measurement rounds, the statistics of a post-measurement state cannot be explained via local hidden variable models. This phenomenon was termed hidden nonlocality and has raised interest in nonlocality in the sequential scenario. A few results have been obtained in this scenario.

It has been shown that if one party performs sequential measurements, often equivalently phrased as multiple Bobs which hand over their system, and the other party, Alice, performs measurements at a single time step, the statistic among Alice and one of the Bobs can violate a Bell inequality \cite{multipleB0,multipleB}. Allowing both parties to perform sequential measurements, the question regarding what correlations should been considered local has been addressed and specific models generalizing the notion of locality to the sequential scenario have been proposed which respect no-signaling among the parties throughout the entire sequence \cite{Gallego2014}. Further, the NPA hierarchy \cite{NPA}, a numerical approach to approximate the quantum set, has been adapted to the sequential setting  with no-signaling \cite{seqNPA} and  it has been noticed that the sequential scenario provides an advantage in the task of randomness certification \cite{seqrand,seqNPA}. Sequential correlations may mimick the presence of higher-dimensional entanglement and witnesses have been derived in order to detect genuine multi-level entanglement \cite{multilevel} and certify the irreducible entanglement dimension  \cite{dim_test_multilevel}.

However, one of the main questions --whether entanglement and nonlocality are equivalent in the sequential scenario-- has not yet been answered. Whereas local filtering reveals the nonlocality of some entangled states, there exist others which neither show nonlocality nor hidden nonlocality despite being entangled \cite{state}. It is not clear whether this remains true if the entire statistics and more than two time steps are considered. 

In this work we shed light on the role of signaling when addressing this question. More precisely, we show that  if signaling between time steps  is allowed there cannot be any advantage beyond hidden nonlocality even if longer sequences are considered and hence entanglement and nonlocality remain distinct features in this scenario. 
This paper is structured as follows. We first introduce the scenario. Then we characterize the spatio-temporal polytope, i.e. the no-signaling polytope in this scenario, by identifying its extreme point. We then provide a notion of locality and  a set of inequalities determining the local polytope for two inputs and two outputs. We further consider a concatenation of these inequalities which results in a (device-dependent)  Schmidt number witness with a considerably larger gap between qubits and ququarts compared to previously known ones. We finally discuss the implications of our results and compare our model to the ones proposed in \cite{Gallego2014} where no-signaling is assumed throughout the sequence.

{\it The scenario.---}We will investigate the correlations arising from sequences of local measurements performed on bipartite or multipartite systems. Considering first the simplest scenario of two measurements steps and two parties, we will be concerned with $ p(a_1a_2b_1b_2\vert x_1x_2y_1y_2)$ which is the probability for obtaining on Alice's  side  in the first time step outcome $a_1$ when measurement $x_1$ is performed  and  in the second time step outcome $a_2$ when measurement $x_2$ is applied and on Bob's  side first $b_1$ and then $b_2$ for the measurement sequence $y_1y_2$.
\begin{figure}[h]
\begin{center}
\includegraphics[width=0.5\columnwidth]{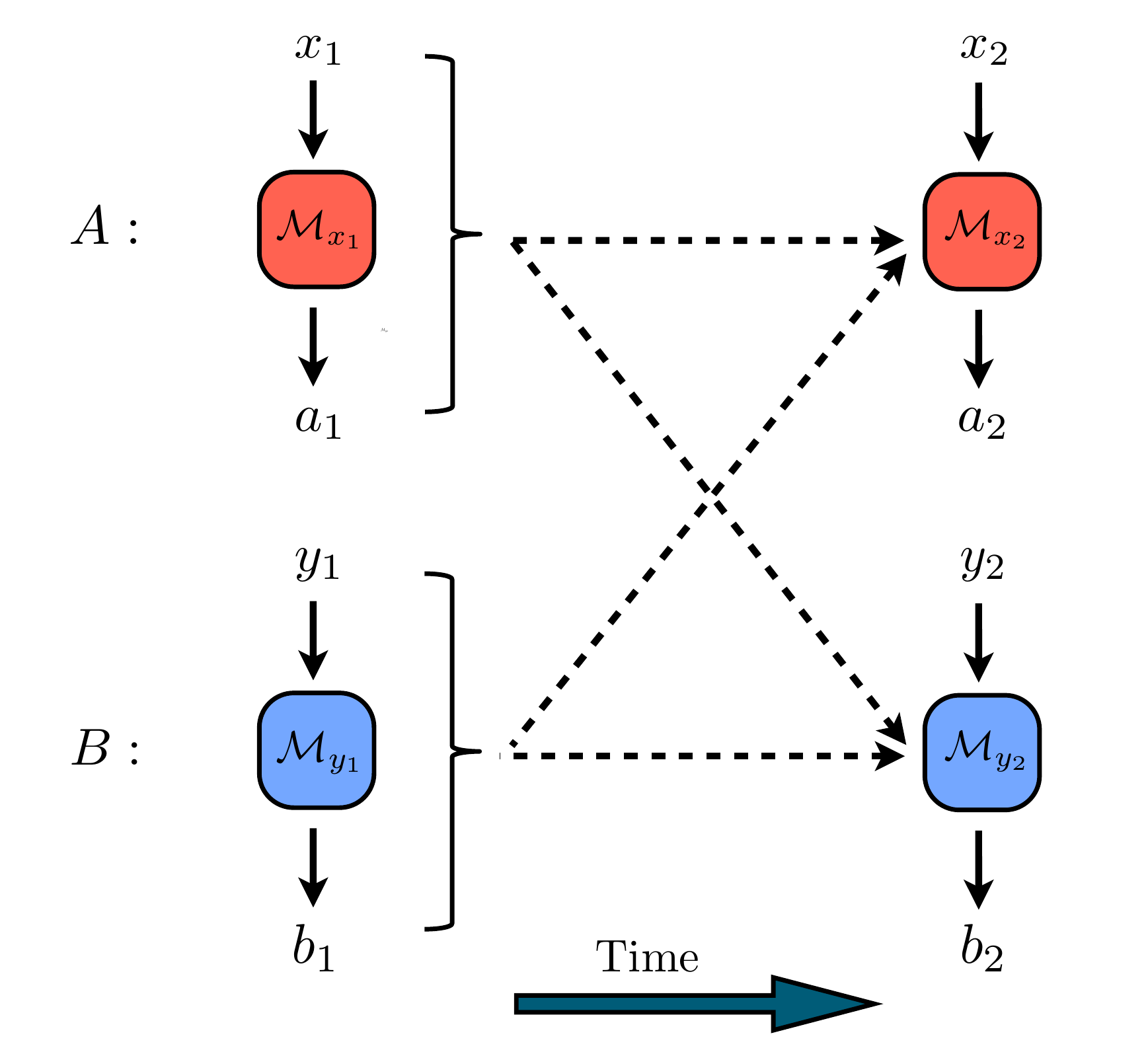} 
\end{center}
\caption{Schematic representation of the scenario considered here. A bipartite quantum state $\rho_{AB}$ is subjected to a  sequence of local measurements $\mathcal{M}_i$. The choices of the measurements for Alice (Bob) are labelled 
by $x_1,x_2$ ($y_1,y_2$) and the outcomes by $a_1,a_2$ ($y_1,y_2$) respectively. The first measurements (choice of setting and outcomes) can influence the second measurements as indicated by the dashed arrows.}
\label{fig:temporal3}
\end{figure}

Due to the sequential character of the measurements the correlations have to  fulfill locally the arrow of time (AoT) constraints \cite{AoT}  (i.e., no-signaling from the future to the past). Moreover, we will assume  the no-signaling condition \cite{PR, Tsirelsonbound} for measurements performed at the same time step. That is the following constraints (as well as the analogous constraints for the other party) have to hold
\begin{align}
\nonumber\sum_{a_2} p(a_1a_2b_1b_2\vert x_1x_2y_1y_2) &= \sum_{a_2} p(a_1a_2b_1b_2\vert x_1x_2'y_1y_2),\\\nonumber
\sum_{a_2,b_2} p(a_1a_2b_1b_2\vert x_1x_2y_1y_2)&= \sum_{a_2, b_2} p(a_1a_2b_1b_2\vert x_1x_2'y_1y_2'),\\
\label{cond}\sum_{a_1,a_2,b_2} p(a_1a_2b_1b_2\vert x_1x_2y_1y_2) &= \sum_{a_1,a_2,b_2} p(a_1a_2b_1b_2\vert x_1'x_2'y_1y_2'),
\end{align}
for all inputs and outputs which do not appear in the summation. 
Note that these constraints define a polytope which can be straightforwardly generalized to longer sequences and more parties and which will be called  the spatio-temporal correlation polytope, denoted by $P_{N,L}^{O,S}$ with $N$ the number of parties, $L$ the length of the sequence, $S$ the number of measurement settings and $O$ the number of outcomes per setting. Note further that the first measurement of Alice can influence the second measurement of Bob (and vice versa) (see Fig. \ref{fig:temporal3}).  This is in contrast to previous works on non-locality in sequential scenarios (for more details see below). The constraints in Eq. (\ref{cond}) have also been considered in the context of quantum key distribution \cite{NS_seq1,NS_seq2} and the randomness of noisy PR correlations \cite{NS_seq3}.

The assumption that signaling is allowed from the past to the future (even to the other party) becomes relevant if, for example, the time between the measurements is too large  such that information can travel from one of the parties to the other within this time window (e.g. if the parties are close, as for example in ion traps). 

{\it The polytopes.---}
In the following we characterize the spatio-temporal correlation polytope in the simplest scenario $P_{2,2}^{2,2}$ by providing its extremal points and a set of inequalities completely describing the corresponding local polytope. Moreover, we investigate its relation to quantum mechanics. Note that all results can be straightforwardly generalized to longer sequences and more parties.

The spatio-temporal polytope  $P_{2,2}^{2,2}$ can be characterized as given in the following Lemma.
\begin{lemma}\label{lemma_polytope}
A correlation is an extreme point of the polytope $P_{2,2}^{2,2}$ if and only if it can be written of the form
\begin{equation}\label{EQextr}
p(a_1a_2b_1b_2\vert x_1x_2y_1y_2)=p_e(a_1b_1|x_1 y_1)p_{e_X}(a_2b_2|X;x_2 y_2)
\end{equation}
 where $X=a_1b_1x_1y_1$, $p_e(a_1b_1|x_1 y_1)$ is an extreme point of the no-signaling polytope and for fixed  $X$ $p_{e_X}(X;x_2 y_2)$ corresponds to an extreme point of the no-signaling polytope (but can be different for different $X$).
\end{lemma}
\begin{proof} Note first that it has been shown that any correlation that fulfills the AoT constraints has to factorize \cite{TCqubit}, i.e. we have that 
\begin{equation}\label{AoT}
p(a_1a_2b_1b_2\vert x_1x_2y_1y_2)=p(a_1b_1|x_1 y_1)p(a_2b_2|X;x_2 y_2)
\end{equation}
 with $p(a_1b_1|x_1 y_1)=\sum_{a_2, b_2}p(a_1a_2b_1b_2\vert x_1x_2y_1y_2)$ and $p(a_2b_2|X;x_2 y_2)=p(a_1a_2b_1b_2\vert x_1x_2y_1y_2)/p(a_1b_1|x_1 y_1)$ for $p(a_1b_1|x_1 y_1)\neq 0$ (otherwise it can be chosen arbitarily). Due to the constraints in Eqs. (\ref{cond}) each factor [$p(a_1b_1|x_1 y_1)$ and $p(a_2b_2|X;x_2 y_2)$ for fixed $X$] has to fulfill the no-signaling conditions.  Note then that if  a correlation is not extremal, i.e, $p(a_1a_2b_1b_2\vert x_1x_2y_1y_2)=\sum q_{i}p_{e_i}(a_1a_2b_1b_2\vert x_1x_2y_1y_2)$, then either $p(a_1b_1|x_1 y_1)=\sum_{a_2, b_2}p(a_1a_2b_1b_2\vert x_1x_2y_1y_2)=\sum q_{i}p_{e_i}(a_1b_1\vert x_1y_1)$ or for all $q_i\neq 0$ it holds that $p_{e_i}(a_1b_1\vert x_1y_1)= p_{e}(a_1b_1\vert x_1y_1)$. However, latter implies that 
\bea \nonumber p(a_2b_2|X;x_2 y_2)&=\sum q_{i}\frac{p_{e_i}(a_1a_2b_1b_2\vert x_1x_2y_1y_2)}{p_{e}(a_1b_1\vert x_1y_1)}\\&=\sum q_{i}p_{e_i}(a_2b_2|X;x_2 y_2)\eea 
and therefore  $ p(a_1a_2b_1b_2\vert x_1x_2y_1y_2)= p_{e}(a_1b_1\vert x_1y_1) \sum q_{i}p_{e_i}(a_2b_2|X;x_2 y_2).$ Hence, we have that  any correlation that is not an extreme point is not of the form given in Eq. (\ref{EQextr}). 

It remains to observe that correlations which are of the form in Eq. (\ref{EQextr}) are the only extremal ones or equivalently that any correlation of the form $p(a_1a_2b_1b_2\vert x_1x_2y_1y_2)=p(a_1b_1|x_1 y_1)p(a_2b_2|X;x_2 y_2)$ where each factor obeys the no-signaling constraints can be achieved by some convex combination of correlations of the form given in Eq.  (\ref{EQextr}). Note first that as the factors in Eq. (\ref{AoT})  correspond to no-signaling correlations they can be represented by $p(a_1b_1|x_1 y_1)=\sum q_{i}p_{e_i}(a_1b_1\vert x_1y_1)$  and $ p(a_2b_2|X;x_2 y_2)= \sum \tilde{q}_{i,X}p_{\tilde{e}_{i,X}}(a_2b_2|X;x_2 y_2)$ where $p_{e_i}(a_1b_1\vert x_1y_1)$ and $p_{\tilde{e}_{i,X}}(a_2b_2|X;x_2 y_2)$ are extreme points of the no-signaling polytope and $\{q_i\}_i$ and $\{\tilde{q}_{i,X}\}_i$ are probability distributions. Using then Eq. (\ref{AoT}) one obtains that any correlation in $P_{2,2}^{2,2}$ is of the form $p(a_1a_2b_1b_2\vert x_1x_2y_1y_2)=\sum q_{i}\tilde{q}_{i,X}p_{e_i}(a_1b_1\vert x_1y_1)p_{\tilde{e}_{i,X}}(a_2b_2|X;x_2 y_2)$ which is a convex combination of the extreme points in Eq. (\ref{EQextr}) and therefore proves the statement.
\end{proof}
Then using the results known for the non-sequential scenario \cite{rev}, one obtains that there are in total $16\times24^4+8 \times 24^8$ extreme points\footnote{This is due to the fact that (see e.g. \cite{rev})  16 extreme points of the non-sequential no-signaling polytope correspond to deterministic assignments (with non-zero $p_e(a_1b_1|x_1y_1)$ for 4 different choices of $X$) and 8 are up to some relabeling equivalent to a PR box (and therefore $p_e(a_1b_1|x_1y_1)$ is non-zero for 8 different choices of $X$). For each $X$  for which $p_e(a_1b_1|x_1y_1)$ is non-zero we can choose a different extreme point $p_{e_X}(a_2b_2|X;x_2y_2)$ and there are 24 different choices for $p_{e_X}(a_2b_2|X;x_2y_2)$, which shows that there are  $16\times24^4+8 \times 24^8$ different extreme points of the polytope $P_{2,2}^{2,2}$.}. Any extreme point for which one of the factors in Eq. (\ref{EQextr})  is a PR box can not be reached within quantum mechanics, as this would require a quantum state and measurements that realize this correlation at the specific point in the measurement history, which do not exist \cite{PR}. Any that is build solely from extreme points of the non-sequential local polytope can  be reached (without the need of entanglement) if one allows the  preparation of an arbitrary separable state at each time step which can depend on the setting and outcomes of previous time steps. We will call in the following the polytope obtained by such extreme points the sequential local polytope, $Q_{N,L}^{O,S}$. 

This definition of the local polytope in our scenario is motivated as follows. First, one would like to recover the well-known notion of (non-sequential) locality in the first time step. Second, given complete information about previous measurement,  the correlations observed in a later time step should be local in the standard (non-sequential) sense. Third, the correlations within this polytope can be obtained from initial separable states and classical shared randomness among the parties and local state preparation (taking into account signaling between  time steps).

The sequential local polytope $Q_{2,2}^{2,2}$ can be characterized by a set of (non-linear) inequalities.
\begin{observation}\label{ineqLocal}
The local polytope $Q_{2,2}^{2,2}$ is determined by the CHSH inequalities (with all possible relabelings) for $ p(a_1b_1\vert x_1y_1)$ and  for $p(a_2b_2|X;x_2 y_2)$ for any $X$, and the trivial constraints (positivity and normalization).
\end{observation}
Recall first that in the corresponding standard Bell scenario the only non-trivial facet inequalities of the local polytope are the CHSH inequalities (with all possible relabelings) \cite{allfacet, allfacets}. 
Then the above observation can be straightforwardly verified by noting that any correlation in $Q_{2,2}^{2,2}$ has to fulfill these inequalities (by being a product of local correlations) and if it fulfills them, then it is in the polytope (which together proves the statement).

Whereas we considered here the simplest sequential scenario, it should be noted that the characterization of the local polytope can be easily generalized to an arbitrary length of the sequence and number of settings and outcomes, even allowing for a different number of settings or outcomes at different time steps. In order to do so note first  that it has been shown in \cite{TCqubit} that the correlations factorize for an arbitrary length of the sequence into correlations of a single time step.
The local polytope can then always be characterized by all the facet inequalities of the local polytopes of the correlations for a fixed time step given previous outcomes and settings. This follows from the argumentation  in the proof of Observation \ref{ineqLocal}. Moreover, note that  the result in \cite{TCqubit}  further implies that the correlations considered here can be written as $
p(a_1a_2\ldots a_Lb_1b_2\ldots b_L\vert x_1x_2\ldots x_Ly_1y_2\ldots y_L)=p(a_1\ldots a_{L-1}b_1\ldots b_{L-1}\vert x_1\ldots x_{L-1}y_1\ldots y_{L-1})p(a_Lb_L|X^{(L)};x_L y_L)
$
where  here and in the following $X^{(L)}=a_1\ldots a_{L-1}b_1\ldots b_{L-1}x_1\ldots x_{L-1}y_1\ldots y_{L-1}$, $p(a_1a_2\ldots a_{L-1}b_1b_2\ldots b_{L-1}\vert x_1x_2\ldots x_{L-1}y_1y_2\ldots y_{L-1})\in P_{N,L-1}^{O,S}$ and $p(a_Lb_L|X^{(L)};x_L y_L)$ fulfills the no-signaling conditions. Using an analogous argumentation as in Lemma \ref{lemma_polytope} one can straightforwardly show that the extreme points are those correlations for which $p(a_1a_2\ldots a_{L-1}b_1b_2\ldots b_{L-1}\vert x_1x_2\ldots x_{L-1}y_1y_2\ldots y_{L-1})$ and $p(a_Lb_L|X^{(L)};x_L y_L)$ are extreme points of their respective polytopes (with possibly different ones for different $X^{(L)}$). Hence, by induction one obtains that also Lemma \ref{lemma_polytope}  can be generalized to arbitarily long sequences.

{\it A device-dependent Schmidt number witness.---}
It is obvious that the quantum bounds for each non-trivial inequality characterizing the local polytope $Q_{2,2}^{2,2}$ are given by the quantum bound for the CHSH inequality, i.e. $2 \sqrt{2}$, and also the (post-measurement) state and measurements attaining the bound are the ones known from the non-sequential case \cite{Tsirelsonbound}. As the  inequalities only depend either on the initial state or one of the post-measurement states, it is meaningful to consider also a concatenated version of the CHSH inequality (which involves sequential correlations) given by 
\begin{align}
\mathcal{B}=\sum_{x_1,y_1=0,1}\sum_{a_1=\pm 1} [p(a_1b_1(a_1,x_1,y_1)\vert x_1 y_1)\sum_{x_2,y_2=0,1}\sum_{a_2,b_2=\pm 1} (-1)^{x_2y_2} a_2 b_2p(a_2b_2|X;x_2 y_2)],\label{concchsh}\end{align}
where here and in the following $b_1(a_1,x_1,y_1)=(-1)^{x_1y_1} a_1$.
Note that we used in the first time step a less common way of writing the CHSH inequality which is related to the standard notion of the CHSH inequality via 
\begin{align}\sum_{x,y=0,1}\sum_{a=\pm 1} p(a [a(-1)^{xy}] \vert x y)=\frac{1}{2}(4+\sum_{x,y=0,1}\sum_{a,b=\pm 1} (-1)^{xy} a b p(ab|xy))\label{CHSHresc}\end{align}
 and therefore using the bounds for the standard CHSH inequality \cite{CHSH, Tsirelsonbound} one obtains that the local bound of the rescaled version of the CHSH inequality in Eq. (\ref{CHSHresc}) is $3$ and the quantum bound is $\sqrt{2}+2$. Obviously, the protocol reaching the quantum bound is the same as for the standard notion of the inequality.

With this the quantum bound of the concatenated CHSH inequality in Eq. (\ref{concchsh}) is $(\sqrt{2}+2)2\sqrt{2}$. This can be shown by using the bounds on the quantities for a single time step (see  \cite{exTC} for an analogous construction). Using that for fixed $X$ the quantum bound of  $\sum_{x_2,y_2=0,1}\sum_{a_2,b_2=\pm 1} (-1)^{x_2y_2} a_2 b_2p(a_2b_2|X;x_2 y_2)$ is $2\sqrt{2}$ and $p(a_1b_1(a_1,x_1,y_1)\vert x_1 y_1)\geq 0$, one obtains 
\begin{align}\nonumber
\mathcal{B}&\underset{\textrm{quantum}}{\leq}2 \sqrt{2}\sum_{x_1,y_1=0,1}\sum_{a_1=\pm 1} p(a_1b_1(a_1,x_1,y_1)\vert x_1 y_1)\\\label{concchshbound}&\underset{\textrm{quantum}}{\leq}2\sqrt{2} (\sqrt{2}+2),\end{align}
where the last inequality follows from the quantum bound of Eq. (\ref{CHSHresc}). In order to reach the maximum one requires no longer a qubit for the maximal violation but a ququart maximally entangled state, where the first measurements are done on one qubit subspace, and the measurements in the second time step on the other. Note that this implies that this quantity cannot be used to detect genuine multi-level entanglement. It is clear that with a qubit the quantum bound cannot be achieved. In order to achieve the maximal violation of the CHSH inequality in the first time step projective measurements are needed, which lead to a separable state to begin with in the second time step, which in turn cannot exceed the local bound. Other strategies might be more beneficial but do not allow to achieve the overall quantum bound.

If one restricts to non-trivial projective measurements, i.e., projective measurements that do not output one measurement outcome with certainty irrespective of the state that is measured, the bound for a qubit can be analogously evaluated from the bounds for the quantities for a single time step. As for qubits the post-measurement state after non-trivial projective measurements is separable, the value of the CHSH inequality in the second time step (for given $a, b, x$ and $y$) cannot exceed the local bound of $2$. Hence, the expression in Eq. (\ref{concchsh}) is bounded by two times the rescaled CHSH inequality for the first time step which results in a bound of $2(\sqrt{2}+2)$ for qubits. In comparison, the bound for ququarts (which is in this case equivalent to the quantum bound) is given by $(\sqrt{2}+2)2\sqrt{2}$. Note that therefore this concatenated CHSH can be employed to witness the Schmidt number \cite{Schmidtnumber} with a  considerable gap between qubits and ququarts. It is device-dependent as we assumed non-trivial projective measurements, however, these do not need to be characterized. Using the CGLMP inequality \cite{CGLMP} a Schmidt number witness has been proposed  in the standard Bell scenario \cite{dimwit}. This witness is device-independent but it requires three-outcome measurements and the gap is small compared to the one proposed here and hence hard to access experimentally. In the considered device-dependent scenario (which requires in this case characterized states) the gap for the nonsequential CGLMP inequality is still small \cite{DDdimwit}, whereas the gap obtained here is straightforwardly accessible in an experiment and requires fewer assumptions. Further, a device-independent dimension test based on the CGLMP inequality (with 4 outcomes) has been implemented experimentally \cite{CGLMPexp} certifying a local dimension of larger than 3. However, also for this test the gap is considerably smaller. It should also be noted that it has been shown recently that there exist states for which the Schmidt number cannot be certified in a device-independent way \cite{SNnotDI} (even in the sequential many-copy case) and hence in general assumptions on the measurements may be necessary anyway.

In order to generalize the Schmidt number witness to longer sequences  (see also \cite{exTC}) one can straightforwardly extend the construction to sequences of arbitrary length $L$ by choosing the rescaled version for the first $L-1$ time steps and the standard notion for the CHSH inequality in the last time step and concatenating them analogously to the one for two time steps leading to a bound of $(\sqrt{2}+2)3^{L-2}2$ for qubits and non-trivial projective measurements. More precisely, one uses as before that in the last $L-1$ time steps the measurements are performed on a separable states and therefore one cannot exceed the local bounds and it is optimal to perform the measurements and use the initial state realizing the quantum bound in the first time step. Note that the quantum bound scales as $(\sqrt{2}+2)^{L-1}2\sqrt{2}$. In order to illustrate this procedure let us consider the case $L=3$. Then the Schmidt number witness is of the form 
\begin{align}\nonumber
\mathcal{B}_{L=3}&=\sum_{x_1,y_1=0,1}\sum_{a_1=\pm 1} \{p(a_1b_1(a_1,x_1,y_1)\vert x_1 y_1)[\sum_{x_2,y_2=0,1}\sum_{a_2=\pm 1} p(a_2b_2(a_2,x_2,y_2)\vert X; x_2 y_2)\\&\sum_{x_3,y_3=0,1}\sum_{a_3,b_3=\pm 1} (-1)^{x_3y_3} a_3 b_3p(a_3b_3|X^{(2)};x_3 y_3)]\}\label{concchsh_L3}\end{align}
and analogously to before the quantum bound is given by 
\begin{align}\nonumber
\mathcal{B}_{L=3}&\underset{\textrm{quantum}}{\leq}2 \sqrt{2}\sum_{x_1,y_1=0,1}\sum_{a_1=\pm 1} p(a_1b_1(a_1,x_1,y_1)\vert x_1 y_1)[\sum_{x_2,y_2=0,1}\sum_{a_2=\pm 1} p(a_2b_2(a_2,x_2,y_2)\vert X; x_2 y_2)]\\
&\underset{\textrm{quantum}}{\leq}\label{concchshbound}2\sqrt{2} (\sqrt{2}+2)\sum_{a_1=\pm 1} p(a_1b_1(a_1,x_1,y_1)\vert x_1 y_1)\underset{\textrm{quantum}}{\leq}2\sqrt{2} (\sqrt{2}+2)^2,\end{align}
where we used here in the first line that for each $X^{(2)}$ the bound of the standard CHSH inequalities is $\sum_{x_3,y_3=0,1}\sum_{a_3,b_3=\pm 1} (-1)^{x_3y_3} a_3 b_3p(a_3b_3|X^{(2)};x_3 y_3)]\leq 2\sqrt{2}$ (and that probabilities are larger than zero) and then twice the quantum bound of Eq. (\ref{CHSHresc}) (together with the positivity of probabilities).
The bound for qubits and non-trivial projective measurements can also be obtained similarly 
\begin{align}\nonumber
\mathcal{B}_{L=3}&\underset{\textrm{d=2 \& proj}}{\leq}2 \sum_{x_1,y_1=0,1}\sum_{a_1=\pm 1} p(a_1b_1(a_1,x_1,y_1)\vert x_1 y_1)[\sum_{x_2,y_2=0,1}\sum_{a_2=\pm 1} p(a_2b_2(a_2,x_2,y_2)\vert X; x_2 y_2)]\\
&\underset{\textrm{d=2 \& proj}}{\leq}\label{concchshbound}6\sum_{a_1=\pm 1} p(a_1b_1(a_1,x_1,y_1)\vert x_1 y_1)\underset{\textrm{d=2 \& proj}}{\leq}6(\sqrt{2}+2).\end{align}
Here we used that after the first measurement all states are separable and hence the measurement statistics of the second and third time step, $ p(a_2b_2(a_2,x_2,y_2)\vert X; x_2 y_2)$ and $p(a_3b_3|X^{(2)};x_3 y_3)$, cannot exceed the local bounds of the respective inequalities. More precisely, we used first together with the positivity of the probabilities that for any $X^{(2)}$ the local bound for the standard CHSH inequality implies that $\sum_{x_3,y_3=0,1}\sum_{a_3,b_3=\pm 1} (-1)^{x_3y_3} a_3 b_3p(a_3b_3|X^{(2)};x_3 y_3)]\leq 2$, then for any $X$ that the local bound is $\sum_{x_2,y_2=0,1}\sum_{a_2=\pm 1} p(a_2b_2(a_2,x_2,y_2)\vert X; x_2 y_2)\leq 3$ and then finally that the maximal attainable value of Eq. (\ref{CHSHresc}) with qubits and projective measurements is $\sqrt{2}+2$.
The bounds for sequences of arbitrary length $L$ can be obtained analogously.

{\it Comparison with the local models of \cite{Gallego2014} and discussion.---}
In \cite{Gallego2014} two different local models have been discussed. One of them is based on a resource theoretic approach called operationally local correlations and the other one is a straightforward generalization of the non-sequential case called time-ordered local model. For two time steps the latter one is given by 

\begin{equation}\label{TOLoc}
p(a_1a_2b_1b_2\vert x_1x_2y_1y_2)=\sum q_\lambda p^\lambda(a_1b_1|x_1 y_1)p^\lambda(a_2b_2|x_2 y_2)
\end{equation}
and is a subset of the operationally local correlations. First it should be noted that both models respect no-signaling among the parties throughout the whole sequence as this is the scenario the authors were interested in. This is in strong contrast to the scenario we consider here. However, in the case of a sequence of length two and a single choice of setting in the first time step for both parties the authors prove that showing no hidden nonlocality and having a time-ordered local model is equivalent. Hence, despite the fact that the constraints are different for this simple scenario our sequential local polytope and the time-ordered model coincide. However, in general this is not the case. The dependence on previous outcomes and settings does not make it possible to write any correlations in $Q_{N,L}^{O,S}$ in the form of Eq. (\ref{TOLoc}) as there exist correlations that do not satisfy the no-signaling correlations throughout all time steps to begin with. Note that the latter also holds true for the operationally local correlations (for the same reason). 

Interestingly, any time-ordered local model is within $Q_{N,L}^{O,S}$. This can be seen by using that it has been shown \cite{Gallego2014} that any post-selection of a time-ordered local model admits a time-ordered local model. Hence, any probability distribution for a given history is local in the standard sense which implies that it lies within $Q_{N,L}^{O,S}$. Note that this implies that any entangled  state that shows no hidden nonlocality but is still nonlocal on the basis of time-ordered local models allows for some measurement statistics that lies outside of the time-ordered model but all its correlations are still inside of $Q_{N,L}^{O,S}$.

Observation \ref{ineqLocal} implies that there exist entangled states for which all correlations lie within the  local polytope $Q_{2,2}^{2,2}$. These are precisely those that are local in the standard Bell scenario and also do not display hidden nonlocality (see \cite{state} for an example), i.e. also post-selection on the first time step does not lead to non-local correlations. Note further that the observation that the correlations of states with no standard nonlocality and no hidden nonlocality lie within the local polytope can be generalized to arbitrary number of settings and outcomes and length of the sequence. Hence, in order to potentially observe sequential nonlocality beyond hidden nonlocality one has to impose no-signaling between the parties for the whole duration of the sequence (or at least for several time steps) as has been done in \cite{Gallego2014}. However,  this may not be naturally the case but has to be enforced in an experimental set-up and restricts the period of time between time steps and possibly also the length of the performed sequences.

{\it Conclusion and Outlook.---}
In this work we studied the correlations arising from sequential local measurements. We first characterized the polytope defined by the AoT constraints and the no-signaling constraints for measurements within the same time step. We then put forward a notion of locality which incorporates the fact that signaling is allowed from earlier to subsequent time steps. We characterized the local polytope in the bipartite scenario with two settings and two outcomes each  by providing a set of inequalities and proposed a device-dependent Schmidt number witness. Finally, we compared our model with the different local models in \cite{Gallego2014} where no-signaling is imposed throughout the whole sequence. Our results show that signaling among subsequent time steps prohibits a system to show nonlocality beyond hidden nonlocality. That is, in this scenario longer sequences do not provide an advantage compared to local filtering  which implies that there do exist entangled states which need to be considered local irrespective of the length of the measurement sequence. 

Whether such states exist if one restricts signaling among the time steps is one of the main open questions. In this context it would be also relevant to study scenarios where the period of time on which previous measurement of Bob can influence Alice's measurements corresponds to several time steps (in order to meet on the one hand the experimental reality and on the other hand still allow for a potential advantage of long sequences). 

If one does not restrict the dimension of the system the quantum set of correlations in the considered sequential scenario can be deduced from the corresponding one for the single time step. This, however, does not hold true if one is interested in the set of correlations which can be obtained from  quantum systems of bounded dimension. It would be interesting to extend the here proposed Schmidt number witness also to the device-independent setting. Moreover, one may also use the dimension-dependence in order to derive semi-device independent self-testing protocols which go beyond the corresponding quantities for a single time step.

\begin{acknowledgements}
I thank Marcus Huber and Otfried G\"uhne for  discussions and Philip Taranto for discussions and comments on the manuscript.
This work has been supported by the Austrian Academy of Sciences and by the Austrian Science Fund (FWF): J 4258-N27 (Erwin-Schr\"odinger Programm) and Y879-N27 (START project).
\end{acknowledgements}

\end{document}